\documentclass{article}

\usepackage{arxiv}
\usepackage{caption}
\usepackage[utf8]{inputenc} 
\usepackage[T1]{fontenc}    
\usepackage{hyperref}       
\usepackage{url}            
\usepackage{booktabs}       
\usepackage{amsfonts}       
\usepackage{nicefrac}       
\usepackage{microtype}      
\usepackage{lipsum}
\usepackage{morefloats}
\usepackage{upgreek}
\usepackage{color}
\usepackage{multirow}
\usepackage{mathrsfs}
\usepackage{amsmath,amssymb,amsthm,graphicx}
\usepackage{dsfont}
\usepackage{enumitem}
\usepackage{eufrak}
\usepackage{hyperref}
\usepackage{arydshln}
\usepackage{lscape}
\newtheoremstyle{thm}
{9pt}
{9pt}
{\itshape}
{}
{\bfseries}
{.}
{ }
{}
\theoremstyle{thm}
\newtheorem{theorem}{Theorem}[section]

\newtheorem{corollary}[theorem]{Corollary}

\newtheoremstyle{def}
{9pt}
{9pt}
{}
{}
{\bfseries}
{.}
{ }
{}
\theoremstyle{def}

\newcommand{\R}{\mathbb{R}} 
\newcommand{\E}{\mathbb{E}} 
\newcommand{\PP}{\mathbb{P}} 

\newcommand{\fse}{\overset{a.s.}{\longrightarrow}}
    \def\cd{\stackrel{\mathcal{D}}{\longrightarrow}}
    \def\cp{\stackrel{\mathcal{\PP}}{\longrightarrow}}

\renewcommand{\footnoterule}{%
	\kern -3.5pt
	\hrule width \textwidth height 1pt
	\kern 3.5pt
}

\makeatletter
\def\blfootnote{\xdef\@thefnmark{}\@footnotetext}
\makeatother

\title{On a new test of fit to the beta distribution}

\author{Bruno Ebner\\
 Institute of Stochastics, \\
Karlsruhe Institute of Technology (KIT), \\
Englerstr. 2, D-76133 Karlsruhe, \\
Germany.\\
\href{mailto:Bruno.Ebner@kit.edu}{Bruno.Ebner@kit.edu}\\
\And
Shawn C. Liebenberg\\
School of Mathematical and Statistical Sciences,\\ North-West University,\\ South Africa. \\
\href{mailto:Shawn.Liebenberg@nwu.ac.za}{Shawn.Liebenberg@nwu.ac.za}
}

\begin{document}

\date{\today}
\maketitle


\begin{abstract}
We propose a new $L^2$-type goodness-of-fit test for the family of beta distributions based on a conditional moment characterisation. The asymptotic null distribution is identified, and since it depends on the underlying parameters, a parametric bootstrap procedure is proposed. Consistency against all alternatives that satisfy a convergence criterion is shown, and a Monte Carlo simulation study indicates that the new procedure outperforms most of the classical tests. Finally, the procedure is applied to a real data set related to air humidity. \end{abstract}

\section{Introduction}
\label{Intro}
The beta distribution is widely used as statistical model in a variety of applied fields, such as meteorology, environmental research, see \cite{F:2004}, geology, see \cite{HW:1989}, communication theory, see \cite{ALT:1990} and traffic flow, see \cite{R:1991}. Due to the support of the distribution on the unit interval $[0,1]$, it is a flexible model for relative frequency data and random probabilities. In Bayesian inference, this family of distributions is used as a conjugate prior probability distribution for binomial and geometric distributions, and it gives its name to the so-called beta wavelets, which are continuous unicycle wavelets of compact support. However, in classical data analysis, before assuming that data stems from a beta distribution, every statistician should check this assumption prior to starting any serious statistical inference. To be precise, we write shorthand ${\cal B}(\alpha,\beta)$, $\alpha,\beta>0$, for the two parameter beta distribution, defined by the density
\begin{equation}\label{eq:dens}
f(x;\alpha,\beta)=\frac{x^{\alpha-1}(1-x)^{\beta-1}}{B(\alpha,\beta)},\quad 0\le x\le1,
\end{equation}
where $B(\alpha,\beta)=\frac{\Gamma(\alpha)\Gamma(\beta)}{\Gamma(\alpha+\beta)}$ is the beta function, and $\Gamma(\cdot)$ denotes the gamma function. Both $\alpha$ and $\beta$ are shape parameters. The corresponding cumulative distribution function is denoted by $F_{(\alpha,\beta)}(t)=B(t;\alpha,\beta)/B(\alpha,\beta)$, $t\in[0,1]$, where $B(\cdot;\alpha,\beta)$ is the incomplete beta function, and we write $\mathscr{B}:=\{{\cal B}(\alpha,\beta):\alpha,\beta>0\}$ for the family of beta distributions. See \cite{JKB:1995}, chapter 25, for more information on the family of beta distributions. To be specific, let $X,X_1,X_2,\ldots$ be independent and identically distributed (i.i.d.) and $[0,1]$-valued random variables defined on a common probability space $(\Omega,\mathcal{A},\mathbb{P})$, and denote the distribution of $X$ by $\mathbb{P}^X$. We test the composite hypothesis
\begin{equation}\label{eq:H0}
H_0:\;\mathbb{P}^X\in\mathscr{B}
\end{equation}
against general alternatives. The literature for this testing problem is hitherto comparably scarce. Apart from classical omnibus procedures, which are based on the empirical distribution function, such as the Kolmogorov-Smirnov test, the Cram\'{e}r-von Mises test or the Anderson-Darling test, only one goodness-of-fit procedure to the beta model has been considered: In \cite{RA:2018} the authors propose a procedure based on a characterisation of the beta distribution via an initial value problem of a second order differential equation, where the Laplace transform of the beta distribution constitutes the unique solution. The test statistic is of weighted $L^2$-type, and it uses the empirical Laplace transform. Note that the results presented in \cite{R:2011} are not reviewed here, since the methodology first uses an estimation procedure to fix the parameters, and it then tests a single hypothesis, in contrast to the composite testing problem treated in the present paper.

This article studies a test procedure for (\ref{eq:H0}) based on a characterisation of the family of beta distributions given in Theorem 1 of \cite{A:1991}. Adapted to the standard two parameter beta distribution, the theorem reads as follows.
\begin{theorem}
Let $X$ be a random variable taking values in $[0,1]$. Then $X\sim {\cal B}(\alpha,\beta)$ for $\alpha,\beta>0$ if and only if
\begin{equation*}
\E(X|X\ge t)=\frac{\alpha}{\alpha+\beta}\left(1+\frac{t^\alpha(1-t)^{\beta}}{\alpha B(\alpha,\beta) \mathbb{P}(X\ge t)}\right),\quad 0\le t\le 1.
\end{equation*}
\end{theorem}
We denote by $\mathbf{1}\{\cdot\}$ the indicator function, with this notation, a short calculation yields the following equivalent characterisation, which is the rationale for the new test statistic.
\begin{corollary}\label{cor:char}
Let $X$ be a random variable taking values in $[0,1]$. Then $X\sim {\cal B}(\alpha,\beta)$ for $\alpha,\beta>0$ if and only if
\begin{equation*}
(\alpha+\beta)\E\left(X\mathbf{1}\{X\ge t\}\right)=\alpha\E\mathbf{1}\{X\ge t\}+\frac{t^\alpha(1-t)^{\beta}}{B(\alpha,\beta)},\quad 0\le t\le 1.
\end{equation*}
\end{corollary}
In the following, we write $\widehat{\alpha}_n,\widehat{\beta}_n$ for consistent estimators of $\alpha$ and $\beta$, i.e., we assume that  $(\widehat{\alpha}_n,\widehat{\beta}_n)\cp(\alpha,\beta)$ as $n\rightarrow\infty$ (such as maximum likelihood estimators, moment estimators, etc.). Here, and in the following $\cp$ denotes convergence in probability. Based on Corollary \ref{cor:char}, we propose the $L^2$-type test statistic
\begin{equation*}
T_{n}=n\int_{0}^1\left|\frac1n\sum_{j=1}^n\left(\left(\widehat{\alpha}_n+\widehat{\beta}_n\right)X_j-\widehat{\alpha}_n\right)\mathbf{1}\{X_j\ge t\}-\frac{t^{\widehat{\alpha}_n}(1-t)^{\widehat{\beta}_n}}{B(\widehat{\alpha}_n,\widehat{\beta}_n)}\right|^2\mbox{d}t.
\end{equation*}
Rejection of $H_0$ will be for large values of $T_n$. By direct calculations, we obtain the numerical stable representation
\begin{eqnarray}
T_{n}&=&\frac1n\sum_{i,j=1}^n\left(\left(\widehat{\alpha}_n+\widehat{\beta}_n\right)^2X_jX_k-\widehat{\alpha}_n\left(\widehat{\alpha}_n+\widehat{\beta}_n\right)(X_j+X_k)+\widehat{\alpha}_n^2\right)\min(X_j,X_k) \nonumber\\
&&-\frac{2B(\widehat{\alpha}_n+1,\widehat{\beta}_n+1)}{B(\widehat{\alpha}_n,\widehat{\beta}_n)}\sum_{j=1}^n\left(\left(\widehat{\alpha}_n+\widehat{\beta}_n\right)X_j-\widehat{\alpha}_n\right)F_{(\widehat{\alpha}_n+1,\widehat{\beta}_n+1)}(X_j)+n\frac{B\left(2\widehat{\alpha}_n+1,2\widehat{\beta}_n+1\right)}{B\left(\widehat{\alpha}_n,\widehat{\beta}_n\right)^2}.
\label{teststatTn}
\end{eqnarray}
The rest of the paper is organized as follows. In Section \ref{sec:AuN} we derive the asymptotic distribution of $T_n$ under the assumption of a beta law using a Hilbert space framework, and we show that the needed parametric bootstrap procedure is well calibrated. Section \ref{sec:cons} provides the consistency of the procedure under a convergence assumption for the estimators. A Monte Carlo simulation study, which is performed in Section \ref{sec:simu} indicates that the new test is a strong competitor to the classical procedures. The procedure is applied to real data sets in Section \ref{sec:real}, and we finish the paper by drawing conclusions and stating open problems in Section \ref{sec:CaO}.

\section{Asymptotics under the null hypothesis}\label{sec:AuN}
In this section we derive the asymptotic distribution under the null hypothesis. Due to the $L^2$-structure of the test statistic, a convenient setting is the separable Hilbert space $\mathbb{H}=L^2([0,1],\mathfrak{B},{\rm d}t)$ of (equivalence classes of) measurable functions $f:[0,1] \rightarrow \R$
satisfying $\int_0^1 |f(t)|^2 \, {\rm d}t < \infty$. Here, $\mathfrak{B}$ denotes the Borel sigma-field on $[0,1]$. The scalar product and the norm in $\mathbb{H}$ will be denoted by
\begin{equation*}
\langle f,g \rangle_{\mathbb{H}} = \int_0^1 f(t)g(t)\, {\rm d}t, \quad \|f\|_{\mathbb{H}} = \langle f,f \rangle_{\mathbb{H}}^{1/2}, \quad f,g \in \mathbb{H},
\end{equation*}
respectively. In view of the bootstrap procedure described later, we consider a slightly more general approach: Let $X_{n,1},\ldots,X_{n,n}$ be a triangular array of rowwise i.i.d. random variables, and suppose $X_{n,1}\sim {\cal B}(\alpha_n,\beta_n)$ for some sequence of parameters $(\alpha_n,\beta_n)\in(0,\infty)^2$, where $\lim_{n\rightarrow\infty}(\alpha_n,\beta_n)=(\alpha_0,\beta_0)$. In the following, let $(\widehat{\alpha}_n,\widehat{\beta}_n)$ be the maximum-likelihood estimators (MLE) of $(\alpha_n,\beta_n)$. Note that we have the linear representation (see \cite{BD:2015}, Section 6.2.1)
\begin{equation*}
\sqrt{n}\left((\widehat{\alpha}_n,\widehat{\beta}_n)^\top-(\alpha_n,\beta_n)^\top\right)=\frac{1}{\sqrt{n}}\sum_{j=1}^n\ell(X_{n,j},\alpha_n,\beta_n)+o_\mathbb{P}(1),
\end{equation*}
where $\ell(x,\alpha,\beta)=I_1(\alpha,\beta)^{-1}\nabla\log(f(x;\alpha,\beta))$, $I_1$ is the Fisher-information matrix, $\nabla\log(f(\cdot))$ is the score vector, $x^\top$ stands for the transpose of a vector $x$, and $o_\mathbb{P}(1)$ denotes a term that converges to 0 in probability. Direct calculations involving the density $f$ in (\ref{eq:dens}) show
\begin{eqnarray*}
\nabla\log(f(x;\alpha,\beta))\!\!\!\!\!&\!=\!&\!\!\!\!\!(\Psi(\alpha+\beta)-\Psi(\alpha)+\log(x),\Psi(\alpha+\beta)-\Psi(\beta)+\log(1-x))^\top,\\
I_1(\alpha,\beta)^{-1}\!\!\!\!\!&=&\!\!\!\!\!\left((\Psi_1(\alpha)+\Psi_1(\beta))\Psi_1(\alpha+\beta)\!-\!\Psi_1(\alpha)\Psi_1(\beta)\right)^{-1}\!\!\left(\begin{array}{cc}\!\!\Psi_1(\alpha+\beta)\!-\!\Psi_1(\beta) & -\Psi_1(\alpha+\beta) \\ -\Psi_1(\alpha+\beta) & \Psi_1(\alpha+\beta)\!-\!\Psi_1(\alpha)\!\!\end{array}\right),
\end{eqnarray*}
where $\Psi(\cdot)$ is the digamma, $\Psi_1(\cdot)$ is the trigamma function, see \cite{JKB:1995}, formula (25.32).
We define $g(t,\alpha,\beta)=\frac{t^\alpha(1-t)^\beta}{B(\alpha,\beta)}$, which is differentiable with respect to both positive parameters, and
\begin{equation*}
Z_n(t)=\frac{1}{\sqrt{n}}\sum_{j=1}^n\left(\left(\widehat{\alpha}_n+\widehat{\beta}_n\right)X_{n,j}-\widehat{\alpha}_n\right)\mathbf{1}\{X_{n,j}\ge t\}-g(t,\widehat{\alpha}_n,\widehat{\beta}_n),\quad 0\le t\le1,
\end{equation*}
as well as the processes
\begin{eqnarray*}
\widetilde{Z}_n(t)&=&\frac{1}{\sqrt{n}}\sum_{j=1}^n\left(\left(\alpha_n+\beta_n\right)X_{n,j}-\alpha_n\right)\mathbf{1}\{X_{n,j}\ge t\}-g(t,\alpha_n,\beta_n)\\&&+\frac1n\sum_{j=1}^n\left(\begin{array}{c}(X_{n,j}-1)\mathbf{1}\{X_{n,j}\ge t\}-\frac{\partial}{\partial\alpha}g(t,\alpha_n,\beta_n)\\X_{n,j}\mathbf{1}\{X_{n,j}\ge t\}-\frac{\partial}{\partial\beta}g(t,\alpha_n,\beta_n)\end{array}\right)^\top\sqrt{n}\left(\begin{array}{c}\widehat{\alpha}_n - \alpha_n \\ \widehat{\beta}_n -\beta_n\end{array}\right),
\end{eqnarray*}
and
\begin{eqnarray*}
\breve{Z}_n(t)&=&\frac{1}{\sqrt{n}}\sum_{j=1}^n\left(\left(\alpha_n+\beta_n\right)X_{n,j}-\alpha_n\right)\mathbf{1}\{X_{n,j}\ge t\}+\ell(X_{n,j};\alpha_n,\beta_n)^\top\Upsilon(t,\alpha_n,\beta_n)-g(t,\alpha_n,\beta_n),
\end{eqnarray*}
where $\Upsilon(t,\alpha_n,\beta_n)=\mathbb{E}\left(\left(X_1-1, X_1 \right)^\top\mathbf{1}\{X_1\ge t\}\right)-\left(\frac{\partial}{\partial\alpha}g(t,\alpha_n,\beta_n),\frac{\partial}{\partial\beta}g(t,\alpha_n,\beta_n)\right)^\top$.
Note that by using a computer algebra system (like Maple or Mathematica), the expectation in $\Upsilon$ can be derived explicitly. In the following, we denote by $\cd$ weak convergence (or alternatively convergence in distribution), whenever random elements (or random variables) are considered.
\begin{theorem}\label{thm:cdZ}
Under the triangular array of row-wise independent $X_{n,j}\sim {\cal B}(\alpha_n,\beta_n)$, there exists a centred Gaussian element $Z$ of $\mathbb{H}$ with covariance kernel
\begin{equation}\label{eq:CovK}
K_Z(s,t)=\mathbb{E}(h_{\alpha_0,\beta_0}(X,s)h_{\alpha_0,\beta_0}(X,t)),\quad s,t\in[0,1],
\end{equation}
with
\begin{equation*}
h_{\alpha,\beta}(X,s)=\left(\left(\alpha+\beta\right)X-\alpha\right)\mathbf{1}\{X\ge s\}+\ell(X;\alpha,\beta)^\top\Upsilon(s,\alpha,\beta)-g(s,\alpha,\beta),\quad s\in[0,1],
\end{equation*}
such that $Z_n\cd Z$ in $\mathbb{H}$ as $n\rightarrow\infty$.
\end{theorem}
\begin{proof} A first order multivariate Taylor expansion at the point $(\alpha_n,\beta_n)$ and the consistency of the MLE estimators $(\widehat{\alpha}_n,\widehat{\beta}_n)$ show $\|Z_n-\widetilde{Z}_n\|_{\mathbb{H}}\cp0$ and an application of the law of large numbers in Hilbert spaces yields $\|\widetilde{Z}_n-\breve{Z}_n\|_{\mathbb{H}}\cp0$ after some calculations. Writing
\begin{equation*}
Z_{n,j}(t)=\left(\left(\alpha_n+\beta_n\right)X_{n,j}-\alpha_n\right)\mathbf{1}\{X_{n,j}\ge t\}+\ell(X_{n,j};\alpha_n,\beta_n)^\top\Upsilon(t,\alpha_n,\beta_n)-g(t,\alpha_n,\beta_n),\quad j=1,\ldots,n,
\end{equation*}
some algebra gives $\mathbb{E}Z_{n,j}=0$. Since $Z_{n,j}$, $j=1,\ldots,n$, are row-wise i.i.d., we have
\begin{equation*}
\lim_{n\rightarrow\infty}\mathbb{E}\|\breve{Z}_n\|^2_{\mathbb{H}}=\lim_{n\rightarrow\infty}\frac1n\sum_{j,k=1}^n\mathbb{E}\langle Z_{n,j},Z_{n,k}\rangle_{\mathbb{H}}=\int_0^1\mathbb{E}Z_{n,1}^2(t)\,\mbox{d}t<\infty.
\end{equation*}
Moreover, we see that $\lim_{n\rightarrow\infty}\mathbb{E}(Z_{n,1}(s)Z_{n,1}(t))=\mathbb{E}(Z(s)Z(t)),$ where $Z(\cdot)$ is the stated Gaussian element in the theorem. For each $\tau\in\mathbb{H}$, the Lindeberg-Feller central limit theorem implies
\begin{equation*}
\frac1{\sqrt{n}}\sum_{j=1}^n\langle Z_{n,j},\tau\rangle_{\mathbb{H}}\cd N(0,\sigma^2_{\alpha_0,\beta_0}(\tau)),
\end{equation*}
where $\sigma^2_{\alpha_0,\beta_0}(\tau)=\lim_{n\rightarrow\infty}\mathbb{E}\langle Z_{n,1},\tau\rangle_{\mathbb{H}}=\mathbb{E}\langle Z,\tau\rangle_{\mathbb{H}}$. Application of Lemma 3.1 in \cite{CW:1998} shows $\breve{Z}_n\cd \mathcal{Z}$ for some Gaussian element $\mathcal{Z}$ of $\mathbb{H}$ with covariance operator $\Sigma_{\alpha_0,\beta_0}$ satisfying $\sigma^2_{\alpha_0,\beta_0}(\tau)=\langle \Sigma_{\alpha_0,\beta_0}\tau,\tau\rangle_{\mathbb{H}}$ for each $\tau\in\mathbb{H}\setminus\{0\}$. Since the kernel $K_Z$ figuring in (\ref{eq:CovK}) satisfies
\begin{equation*}
\sigma^2_{\alpha_0,\beta_0}(\tau)=\int_0^1\int_0^1K_Z(s,t)\tau(s)\tau(t)\,\mbox{d}s\,\mbox{d}t,
\end{equation*}
the claim follows.
\end{proof}
By an application of the continuous mapping theorem, we have the following corollary.
\begin{corollary}\label{cor:nH}
Under the same assumptions as in Theorem \ref{thm:cdZ}, we have $T_{n}\cd \|Z\|_{\mathbb{H}}^2$ as $n\rightarrow\infty$.
\end{corollary}

The distribution of $\|Z\|_{\mathbb{H}}^2$ is known to have the equivalent representation $\sum_{j=1}^\infty\lambda_j(\alpha_0,\beta_0)N_j^2$, where $N_1,N_2,\ldots$ are independent, standard normally distributed random variables, and $\lambda_1(\alpha_0,\beta_0),\lambda_2(\alpha_0,\beta_0),\ldots$ is the decreasing series of non-zero eigenvalues of the integral operator
\begin{equation*}
\mathcal{K}:\mathbb{H}\rightarrow\mathbb{H},\quad f\mapsto\mathcal{K}f(\cdot)=\int_0^1K_Z(\cdot,t)f(t)\mbox{d}t.
\end{equation*}
Clearly, this operator depends on the parameters $(\alpha_0,\beta_0)$. To calculate the eigenvalues $\lambda$ of $\mathcal{K}$, one has to solve the homogeneous Fredholm integral equation of the second kind
\begin{equation}\label{eq:inteq}
\int_0^1 K_Z(x,t)f(t)\mbox{d}t=\lambda f(x),\quad 0\le x\le1,
\end{equation}
see, e.g., \cite{KS:1947}. Due to the complexity of the covariance kernel, it seems hopeless to find explicit solutions of (\ref{eq:inteq}) and hence formulae for the eigenvalues. Furthermore, since the true parameters $(\alpha_0,\beta_0)$ are unknown in practice, the limiting null distribution cannot be used to derive critical values of the test. A solution to this problem is provided by a parametric bootstrap procedure as suggested in \cite{gurtler2000} and which is stated as follows:
\begin{enumerate}
  \item[(1)] Compute $\big(\widehat{\alpha}_n,\widehat{\beta}_n\big)=\big(\widehat{\alpha}_n(X_1,\ldots,X_n),\widehat{\beta}_n(X_1,\ldots,X_n)\big)$.
  \item[(2)] Conditionally on $\big(\widehat{\alpha}_n,\widehat{\beta}_n\big)$ simulate $B$ bootstrap samples $X_{j,1}^*,\ldots,X_{j,n}^*$, i.i.d. from $\mathcal{B}\left(\widehat{\alpha}_n,\widehat{\beta}_n\right)$, and compute $T_{n,j}^*=T_n(X_{j,1}^*,\ldots,X_{j,n}^*)$, $j=1,\ldots,B$.
  \item[(3)] Derive an empirical $(1-\upalpha)$-quantile of $c_{n,B}^*(\upalpha)$ of $T_{n,1}^*,\ldots,T_{n,B}^*$.
  \item[(4)] Reject the hypothesis (\ref{eq:H0}) at level $\upalpha$ if $T_n(X_1,\ldots,X_n)>c_{n,B}^*(\upalpha)$.
\end{enumerate}
Note that for each computation of $T_{n,j}^*$, parameter estimation has to be done separately for each $j$.
Following the notation and methodology of \cite{H:1996} we prove that this bootstrap test has asymptotic level $\upalpha$ as $n,B\rightarrow\infty$. Denote the distribution function of $T_n$ under $\mathcal{B}(\alpha_n,\beta_n)$ by
\begin{equation*}
H_n^{(\alpha_n,\beta_n)}(t)=\mathbb{P}_{(\alpha_n,\beta_n)}(T_n\le t),\quad t>0,
\end{equation*}
and write $H^{(\alpha_0,\beta_0)}(\cdot)$ for the distribution of $\|Z\|_{\mathbb{H}}^2$. Note that $H^{(\alpha_0,\beta_0)}$ is continuous and strictly increasing on $\{t>0: 0< H^{(\alpha_0,\beta_0)}(t) < 1\}$. By Corollary \ref{cor:nH} it holds that $H_n^{(\alpha_n,\beta_n)}(t) \rightarrow H^{(\alpha_0,\beta_0)}(t)$ for each $t>0$ as $n\rightarrow\infty$, so by continuity of $H^{(\alpha_0,\beta_0)}$ we have
\begin{equation*}
\sup_{t>0}\left|H_n^{(\alpha_n,\beta_n)}(t)-H^{(\alpha_0,\beta_0)}(t)\right|\longrightarrow0\quad \mbox{as}\,n\rightarrow\infty.
\end{equation*}
A combination of the last result with the consistency of the MLE $\left(\widehat{\alpha}_n,\widehat{\beta}_n\right)$ yields
\begin{equation*}
\sup_{t>0}\left|H_n^{(\widehat{\alpha}_n,\widehat{\beta}_n)}(t)-H^{(\alpha_0,\beta_0)}(t)\right|\cp0\quad \mbox{as}\,n\rightarrow\infty.
\end{equation*}
Hence, with $\widehat{H}_{n,B}(t)=\frac{1}{B}\sum_{j=1}^B\mathbf{1}\{T_{n,j}^*\le t\}$ denoting the empirical distribution function of $T_{n,1}^*,\ldots,T_{n,B}^*$, we have by an identical construction as in (3.10) of \cite{H:1996}
\begin{equation*}
\sup_{t>0}\left|\widehat{H}_{n,B}(t)-H^{(\alpha_0,\beta_0)}(t)\right|\cp0\quad \mbox{as}\,n,B\rightarrow\infty,
\end{equation*}
from which $c_{n,B}^*(\upalpha)\cp \inf\{t:H^{(\alpha_0,\beta_0)}(t)\ge 1-\upalpha\}$ as $n,B\rightarrow\infty$ follows. This implies that if $X_1,\ldots,X_n$ is a random sample from $\mathcal{B}(\alpha_0,\beta_0)$, we have
\begin{equation*}
\lim_{n,B\rightarrow\infty}\mathbb{P}_{(\alpha_0,\beta_0)}(T_n(X_1,\ldots,X_n)> c_{n,B}^*(\upalpha))=\upalpha,
\end{equation*}
ensuring an asymptotic level $\upalpha$ test.

\section{Consistency against alternatives}\label{sec:cons}
In this section, let $X$ be a random variable taking values in $[0,1]$, and let $X_1,\ldots,X_n,\ldots$ be i.i.d. copies of $X$. We assume that 
\begin{equation}\label{eq:convestAlt}
\big(\widehat{\alpha}_n,\widehat{\beta}_n\big)\cp (\alpha,\beta)\quad\mbox{as $n\rightarrow\infty$}
\end{equation}
for some $\alpha,\beta>0$. Notice that, since the random variables have compact support, each moment of $X$ is finite. In the next theorem $\fse$ stands for $\mathbb{P}$-almost sure convergence.
\begin{theorem}\label{thm:cons}
Under the stated assumptions, we have
\begin{equation*}
\frac{T_n}{n}\fse{\rm \Delta}_{\alpha,\beta}=\int_0^1\left|(\alpha+\beta)\E\left(X\mathbf{1}\{X\ge t\}\right)-\alpha\mathbb{P}(X\ge t)-\frac{t^\alpha(1-t)^{\beta}}{B(\alpha,\beta)}\right|^2\,{\rm d}t\quad \mbox{as}\;n\rightarrow\infty.
\end{equation*}
\end{theorem}
\begin{proof} Notice that, using exactly the same arguments as in the proof of Theorem \ref{thm:cdZ}, we have $n^{-1/2}\|Z_n-\widetilde{Z}_n\|_{\mathbb{H}}\fse0$ and $n^{-1/2}\|\widetilde{Z}_n-\breve{Z}_n\|_{\mathbb{H}}\fse0$ as $n\rightarrow\infty$. Hence, by the triangle inequality we also have $n^{-1/2}\|Z_n-\breve{Z}_n\|_{\mathbb{H}}\fse0$. By the law of large numbers in Hilbert space, it is easy to see that $n^{-1/2}\max(\|Z_n\|_{\mathbb{H}},\|\breve{Z}_n\|_{\mathbb{H}})$ converges in probability to a finite constant. Consequently the reverse triangle inequality implies
\begin{equation*}
n^{-1}\left|T_n-\|\breve{Z}_n\|_{\mathbb{H}}^2\right|=n^{-1}\left|\|Z_n\|_{\mathbb{H}}-\|\breve{Z}_n\|_{\mathbb{H}}\right|\left|\|Z_n\|_{\mathbb{H}}+\|\breve{Z}_n\|_{\mathbb{H}}\right|\le 2n^{-1}\|Z_n-\breve{Z}_n\|_{\mathbb{H}}\max(\|Z_n\|_{\mathbb{H}},\|\breve{Z}_n\|_{\mathbb{H}})\fse0.
\end{equation*}
By the law of large numbers in Hilbert space and the continuous mapping theorem, we have $n^{-1}\|\breve{Z}_n\|_{\mathbb{H}}^2\fse{\rm \Delta}_{\alpha,\beta}$ and the assertion follows.
\end{proof}
In view of Corollary \ref{cor:char} $\Delta_{\alpha,\beta}=0$ holds if and only if $X$ follows the beta law, and $\Delta_{\alpha,\beta}>0$ otherwise. From the proof of the bootstrap procedure, for a given nominal level $\upalpha$ and the bootstrap critical values $c_{n,B}^*(\upalpha)$ as in Section \ref{sec:AuN} we have $c_{n,B}^*(\upalpha)\cp c$ for $n,B\rightarrow\infty$, for some fixed positive value $c<\infty$. A direct consequence of Theorem \ref{thm:cons} is
\begin{equation*}
\mathbb{P}(T_n>c_{n,B}^*(\upalpha))\rightarrow1\quad\mbox{as}\;n,B\rightarrow\infty,
\end{equation*}
which implies consistency of $T_n$ against any alternative distribution $\mathbb{P}^X$ for which (\ref{eq:convestAlt}) is satisfied.

\section{Simulations}\label{sec:simu}
In this section we investigate the finite-sample performance of our test given in (\ref{teststatTn}), when compared to competitive tests against general alternatives with the use of a Monte Carlo study. A significance level of 10\% ($\upalpha = 0.1$) is used throughout the simulation study, and all calculations were performed using the statistical computing environment \texttt{R}, \cite{RCT2013}. The competitive tests considered in this study are the classical tests based on the empirical distribution function, i.e., the Kolmogorov-Smirnov ($KS_n$), Cramer-von Mises ($CM_n$) and Anderson-Darling ($AD_n$) tests (see \cite{dagostino1986}, chapter 4,  for a discussion on these tests), as well as the test based on the Laplace transform by \cite{RA:2018}. The approach introduced in \cite{RA:2018} considers the differential equation,
$({t}/{\alpha}) L^{\prime \prime}(t)+\left({(\alpha+\beta)}/{\alpha}+({t}/{\alpha})\right) L^{\prime}(t)+L(t)=0$
with the Laplace transform, $L(t)=\E(\exp (-t X))$ of the beta distribution and boundary conditions $L(0) = 1$, $L^{\prime}(0)={-\alpha}/{(\alpha+\beta)}$ to construct the test statistic
\begin{equation*}\begin{aligned}
R F_{n, a}:=\frac{1}{n} \sum_{i,j=1}^{n}\biggl(\frac{2\left[X_{i} X_{j}\left(X_{j}\left(X_{i}-2\right)+1\right)\right] / \widehat{\alpha}_{n}^{2}}{\left(X_{i}+X_{j}+a\right)^{3}}+& \frac{2\left[X_{i} X_{j}\left(1-X_{j}\right)-X_{i} \bar{X}_{n}\left(1-X_{i}\right)\right] / \widehat{\alpha}_{n}^{2} \bar{X}_{n}}{\left(X_{i}+X_{j}+a\right)^{2}} \\
&+\frac{\left[X_{i} X_{j}-2 X_{i} \bar{X}_{n}+\bar{X}_{n}^{2}\right] / \bar{X}_{n}^{2}}{\left(X_{i}+X_{j}+a\right)}\biggr),
\end{aligned}\end{equation*}
where $\bar{X}_{n}=\widehat{\alpha}_{n}/(\widehat{\alpha}_{n}+\widehat{\beta}_{n})$ and $\widehat{\alpha}_{n}$, $\widehat{\beta}_{n}$ are consistent estimators. The tuning parameter $a>0$ provides the above test with a degree of flexibility in terms of power by making it either more sensitive to departures from the beta distribution around zero or near 1 of the distribution. \\
Since the limit null distribution of the tests depend on the shape parameters $\alpha$ and $\beta$, the critical values are obtained with the parametric bootstrap as stated in Section \ref{sec:AuN}, and with the given algorithm.
For the simulation study, the maximum likelihood estimates ($\widehat{\alpha}_{n}$,$\widehat{\beta}_{n}$) are calculated in each Monte Carlo replication with the \texttt{EnvStats} package (see, \cite{EnvStats2013}), which solves the equations,
\begin{equation*}\begin{array}{c}
\Psi(\widehat{\alpha}_n)-\Psi(\widehat{\alpha}_n+\widehat{\beta}_n)= \frac{1}{n} \sum_{i=1}^{n} \log \left(X_{i}\right), \\
\Psi(\widehat{\beta}_n)-\Psi(\widehat{\alpha}_n+\widehat{\beta}_n)= \frac{1}{n} \sum_{i=1}^{n} \log \left(1-X_{i}\right),
\end{array}\end{equation*}
simultaneously. 
For each replication, $B=500$ bootstrap samples are generated. To ensure a more accurate empirical level at a reduced number of bootstrap replication we use the modified critical values of \cite{gurtler2000} given by
\begin{equation*}
\tilde{c}_{n,B}:=T_{(B-[0.1(B+1)])}^{*}+0.90\left(T_{(B-[0.1(B+1)]+1)}^{*}-T_{(B-[0.1(B+1)])}^{*}\right) = T_{(450)}^{*}+0.90\left(T_{(451)}^{*}-T_{(450)}^{*}\right),
\end{equation*}
where $T_{n,(j)}^{*}$ denotes the $j$th order statistic of the bootstrap sample of values of the test statistic $T_{n,1}^{*}, T_{n,2}^{*}, \dots, T_{n,B}^{*}$. This modification is made to $c_{n,B}^*(\upalpha)$ in step 3 of the parametric bootstrap algorithm. \\ \\
The alternative distributions that are used in the simulation study are defined as follows:
\begin{itemize}
\item The truncated normal distribution ${TN}(\mu, \sigma^2)$ with density given by
\begin{equation*}
f_{TN}(\mu, \sigma)=\frac{\exp \left(-\frac{1}{2}\left(\frac{x-\mu}{\sigma}\right)^{2}\right)}{\sigma \sqrt{2 \pi}\left(\Phi\left(\frac{1-\mu}{\sigma}\right)-\Phi\left(-\frac{\mu}{\sigma}\right)\right)}, \quad 0<x<1,
\end{equation*}
where $\Phi(x)=\frac{1}{2}(1+\operatorname{erf}(x / \sqrt{2})$ is the normal distribution function.
\item The composite distribution $BN\left(p, \alpha, \beta, \mu, \sigma^{2}\right)$ consisting of the beta distribution, ${\cal B}(\alpha,\beta)$, and the truncated normal distribution, $TN(\mu, \sigma^2)$. That is,
\begin{equation*}
BN\left(p, \alpha, \beta, \mu, \sigma^{2}\right):=\mathbf{1}_{\{X \leq p\}} X_{1}+\mathbf{1}_{\{X>p\}} X_{2}, \quad 0<p<1,
\end{equation*}
with $X\sim U(0,1)$, $X_1 \sim {\cal B}(\alpha,\beta)$, $X_2 \sim TN(\mu, \sigma^2)$ ($X$, $X_1$ and $X_2$ are independent) and $U(0,1)$ is the standard uniform distribution.
\item The Logit-normal distribution $LT(\mu, \sigma^2)$ with density given by
\begin{equation*}
f_{LT}(\mu, \sigma)=\frac{1}{\sigma \sqrt{2 \pi} x(1-x)} \exp \left(-\frac{(\operatorname{logit}(x)-\mu)^{2}}{2 \sigma^{2}}\right), \quad 0 \leq x \leq 1,
\end{equation*}
where $\operatorname{logit}(x)=\log \{x /(1-x)\}$.
\item The families of distributions $F\circ G$ of random variables generated by the following construction: For an i.i.d. sample $X_1,\ldots,X_n$ with $X_1\sim G$, generate the i.i.d. sample $F(X_1),\ldots,F(X_n)$ with the selection of distributions given in Table \ref{tab:disttable}.
\begin{table}[b]
  \centering
  \caption{Distributions for the construction of $F\circ G$}
    \begin{tabular}{|p{4em}|l|p{4em}|l|c|}
    \hline
    \multicolumn{1}{|c|}{Distribution} & \multicolumn{1}{c|}{$F(x)$} & \multicolumn{1}{c|}{Distribution} & \multicolumn{1}{c|}{$G(x)$} & $F\circ G$ \\
    \hline
    Cauchy\newline{}$C(\theta)$ & $\frac{1}{\pi} \arctan\left(\frac{x}{\theta}\right)+\frac{1}{2}$      & Gompertz\newline{}$GO(\eta,\nu)$ &$1-\exp \left(-\eta \left(e^{{\nu x}}-1\right)\right)$ & $C(\theta)\circ GO(\eta,\nu)$  \\
    exponential\newline{}$EXP(\lambda)$ & $\displaystyle 1-e^{-\lambda x}$ & halfnormal\newline{}$HN(\sigma)$ & ${\displaystyle \operatorname {erf} \left({\frac {x}{\sigma {\sqrt {2}}}}\right)}$ & $EXP(\lambda)\circ HN(\sigma)$ \\
    normal\newline{}$\Phi(\mu,\sigma^2)$ & ${\displaystyle {\frac {1}{2}}\left[1+\operatorname {erf} \left({\frac {x-\mu }{\sigma {\sqrt {2}}}}\right)\right]}$ & Laplace\newline{}$L(\mu,\vartheta)$ & ${\displaystyle {\begin{cases}{\frac {1}{2}}\exp \left({\frac {x-\mu }{\vartheta}}\right){\text{if }}x\leq \mu \\[8pt]1-{\frac {1}{2}}\exp \left(-{\frac {x-\mu }{\vartheta}}\right){\text{if }}x\geq \mu \end{cases}}}$ & $\Phi(\mu,\sigma^2)\circ L(\mu,\vartheta)$ \\
    \hline
    \end{tabular}%
  \label{tab:disttable}%
\end{table}%
%
\end{itemize}
Empirical rejection rates are calculated for sample sizes $n=50$ and $n=100$ using $10 \ 000$ independent Monte Carlo replications. Table \ref{tab:powern50} and \ref{tab:powern100} contain the estimated powers in percentages (rounded to the nearest integer) for the new and competitive tests for a sample size of $n=50$ (Table \ref{tab:powern50}) and $n=100$ (Table \ref{tab:powern100}) against each of the alternative distribution. For ease of comparison, we highlighted the highest power for each alternative.
\begin{table}[t]
  \centering
  \caption{Estimated powers for $n = 50$, significance level of $10\%$ and 10000 replications}
    \begin{tabular}{lccccccc}
    \hline
          & $KS_n$ & $CM_n$ & $AD_n$ & $RF_{n,0.25}$ & $RF_{n,2}$ & $RF_{n,5}$ & $T_n$ \\
    \hline
    ${\cal B}(0.5,0.5)$     & 10    & 10    & 10    & 9    & 10    & 11    & 9 \\ 
    ${\cal B}(1,1)$     & 10    & 10    & 10    & 9     & 11    & 10     & 9 \\ 
    ${\cal B}(2,2)$     & 10    & 10    & 10    & 10     & 9     & 10     & 8 \\ 
    ${\cal B}(0.5,1.5)$     & 10    & 10    & 10    & 9    & 10    & 10    & 9 \\ 
    ${\cal B}(0.5,3)$     & 10    & 10    & 10    & 11    & 10    & 10    & 9 \\ 
    ${\cal B}(1.5,0.5)$     & 10    & 10    & 10    & 9     & 9    & 9    & 9 \\ 
    ${\cal B}(3,0.5)$     & 10    & 10    & 11    & 8     & 8     & 8    & 10 \\ 
    $BN(0.25,0.5,0.5,0.25,0.25)$     & 55    & 61    & 64    & 54    & 51    & 56    & \textbf{71} \\ 
    $BN(0.5,0.5,0.5,0.25,0.25)$     & 49    & 55    & 55    & 45    & 41    & 46   & \textbf{62} \\ 
    $BN(0.75,0.5,0.5,0.25,0.25)$    & 24    & \textbf{27} & 26    & 21    & 21    & 23    & \textbf{27} \\ 
    $BN(0.25,2,2,0.25,0.25)$    & 14    & 15    & 16    & \textbf{17}    & 16    & 16    & 14 \\ 
    $BN(0.25,1.5,0.5,0.25,0.25)$    & 75    & 80    & 81    & 74    & 78    & 81    & \textbf{85} \\ 
    $TN(0.25,0.25)$    & 17    & 19    & \textbf{21}    & 22    & 22    & \textbf{21}    & 17 \\ 
    $TN(0.5,0.25)$    & 19    & 22    & 24    & \textbf{25}    & 24    & 23    & 22 \\ 
    $TN(0.25,0.5)$    & 12    & 13    & 14    & 12    & 14    & \textbf{15}    & \textbf{15} \\ 
    $LT(3,2)$    & 70    & 78    & \textbf{80} & 63    & 64    & 65    & 76 \\ 
    $LT(1,2)$    & 22    & 25    & 25    & 16    & 13    & 15    & \textbf{26} \\ 
    $LT(0.5,3)$    & 17    & 18    & \textbf{20} & 16 &	7 &	7 & 19 \\ 
    $C(1)\circ GO(2,1)$    &  63   &  72   &  85   &  69  & 69  &   67  &  \textbf{89} \\ 
    $EXP(1)\circ HN(1)$    & 26    & 30    & 31 & 34    & 36    & \textbf{38}    & 31 \\ 
    $\Phi(0,1)\circ L(2,0.5)$    & 70    & 75    & 75    & 70    & 71    & 71    & \textbf{77} \\ 
    \hline
    \end{tabular}%
  \label{tab:powern50}%
\end{table}%
\begin{table}[t]
  \centering
  \caption{Estimated powers for $n = 100$, significance level of $10\%$ and 10000 replications}
    \begin{tabular}{lccccccc}
    \hline
          & $KS_n$ & $CM_n$ & $AD_n$ & $RF_{n,0.25}$ & $RF_{n,2}$ & $RF_{n,5}$ & $T_n$ \\
    \hline
        ${\cal B}(0.5,0.5)$      & 11    & 10    & 10    & 9     & 10    & 11    & 10 \\ 
        ${\cal B}(1,1)$      & 11    & 10    & 10    & 9     & 10    & 11    & 9 \\ 
        ${\cal B}(2,2)$      & 10    & 10    & 10    & 9    & 10    & 10     & 9 \\ 
        ${\cal B}(0.5,1.5)$      & 10    & 10    & 10    & 10     & 10    & 10    & 9 \\ 
        ${\cal B}(0.5,3)$      & 10    & 9     & 10    & 10    & 10    & 10    & 10 \\ 
        ${\cal B}(1.5,0.5)$      & 10    & 10    & 10    & 10    & 9    & 10    & 10 \\ 
        ${\cal B}(3,0.5)$      & 10    & 10    & 11    & 9     & 9    & 9    & 10 \\ 
        $BN(0.25,0.5,0.5,0.25,0.25)$      & 79    & 86    & 88    & 80    & 69    & 74    & \textbf{93} \\ 
        $BN(0.5,0.5,0.5,0.25,0.25)$      & 76    & 82    & 82    & 72    & 60    & 67    & \textbf{88} \\
        $BN(0.75,0.5,0.5,0.25,0.25)$     & 38    & 42    & 41    & 34    & 28    & 33    & \textbf{45} \\
        $BN(0.25,2,2,0.25,0.25)$     & 16    & 20    & \textbf{21}    & \textbf{21}    & 18 & 17    & 19 \\ 
        $BN(0.25,1.5,0.5,0.25,0.25)$     & 95    & 98    & 98    & 95    & 96    & 97    & \textbf{99} \\
        $TN(0.25,0.25)$     & 22    & 26    & 28    & \textbf{30}    & 27 & 25    & 25 \\ 
        $TN(0.5,0.25)$     & 26    & 31    & \textbf{33}    & 32    & 27    & 25 & 32 \\ 
        $TN(0.25,0.5)$     & 15    & 16    & 16    & 14    & 16    & 17    & \textbf{19} \\ 
        $LT(3,2)$     & 93    & \textbf{98} & \textbf{98} & 89    & 90    & 91    & 96 \\ 
        $LT(1,2)$     & 37    & 42    & 45    & 26    & 21    & 24 & \textbf{48} \\ 
        $LT(0.5,3)$     & 26    & 31    & \textbf{36}    & 28    & 7 & 7    & 34 \\ 
        $C(1)\circ GO(2,1)$  & 92 & 98 & 99 & 96   &  95  &  95  & \textbf{100} \\
        $EXP(1)\circ HN(1)$     & 43    & 50    & 51    & 58    & 59    & \textbf{61}    & 55 \\ 
        $\Phi(0,1)\circ L(2,0.5)$     & 92    & \textbf{95} & \textbf{95} & 93    & 92    & 93    & \textbf{95} \\
    \hline
    \end{tabular}%
  \label{tab:powern100}%
\end{table}%
It is clear that each of the tests achieve the nominal significance level. The newly proposed test $T_n$, performs better than the classical tests except for a few selected cases, where it is still a strong competitor. The best test for the composite distributions is $T_n$ by a fair margin, except for $BN(0.25,2,2,0.25,0.25)$, where $RF_{n,0.25}$ exhibits a slightly better performance for $n=50$ and matches the $AD_n$ test for $n=100$. The test based on the empirical Laplace transform $RF_{n,a}$ is more competitive for the truncated distributions, except for $TN(0.25,0.5)$, where $T_n$ outperforms its counterparts for $n=100$.  It is evident that the performance of $RF_{n,a}$, is largely influenced by the choice of the tuning parameter. In this regard, our new test also performs better for a variety of the considered alternatives, and it has the added advantage of not being dependent on the choice of a tuning parameter. Overall, the newly proposed test shows very good results, and it proves to be the superior choice in most cases. It is clear that there is no 'best' test, insofar as no procedure can outperform the other procedures uniformly, for details see \cite{J:2000}. Since it is known that the power of the tests can depend on the method of estimation, see \cite{DKO:1990}, we performed a simulation study with the same simulation design but replacing the MLE by the moment estimators of $(\alpha,\beta)$. All the procedures loose power by this estimation method and hence the results are dominated by the empirical powers in Tables \ref{tab:powern50} and \ref{tab:powern100}, so we decided not to state them here.

\section{Real data examples}\label{sec:real}
For a practical application we use the data presented in \cite{R:2011} from the Haarweg Wageningen weather station of the relative humidity of air in May 2007 and 2008. The data is given below for convenience. Figure \ref{fig:betafit} gives an indication of how the data fits a beta distribution. We see from the histogram and density plots in the top panels that the data of May 2007 fits a beta distribution relatively well, whereas the data of May 2008 does not. Here the density plots in red are the distributions ${\cal B}(6.356,1.970)$ and ${\cal B}(2.803,1.456)$, where the parameters have been estimated from the data with the maximum likelihood method. The Q-Q plots in the bottom panels of Figure \ref{fig:betafit} further strengthens this conclusion.
	\begin{figure}[t]
	\centering
		\includegraphics[width=0.70\textwidth]{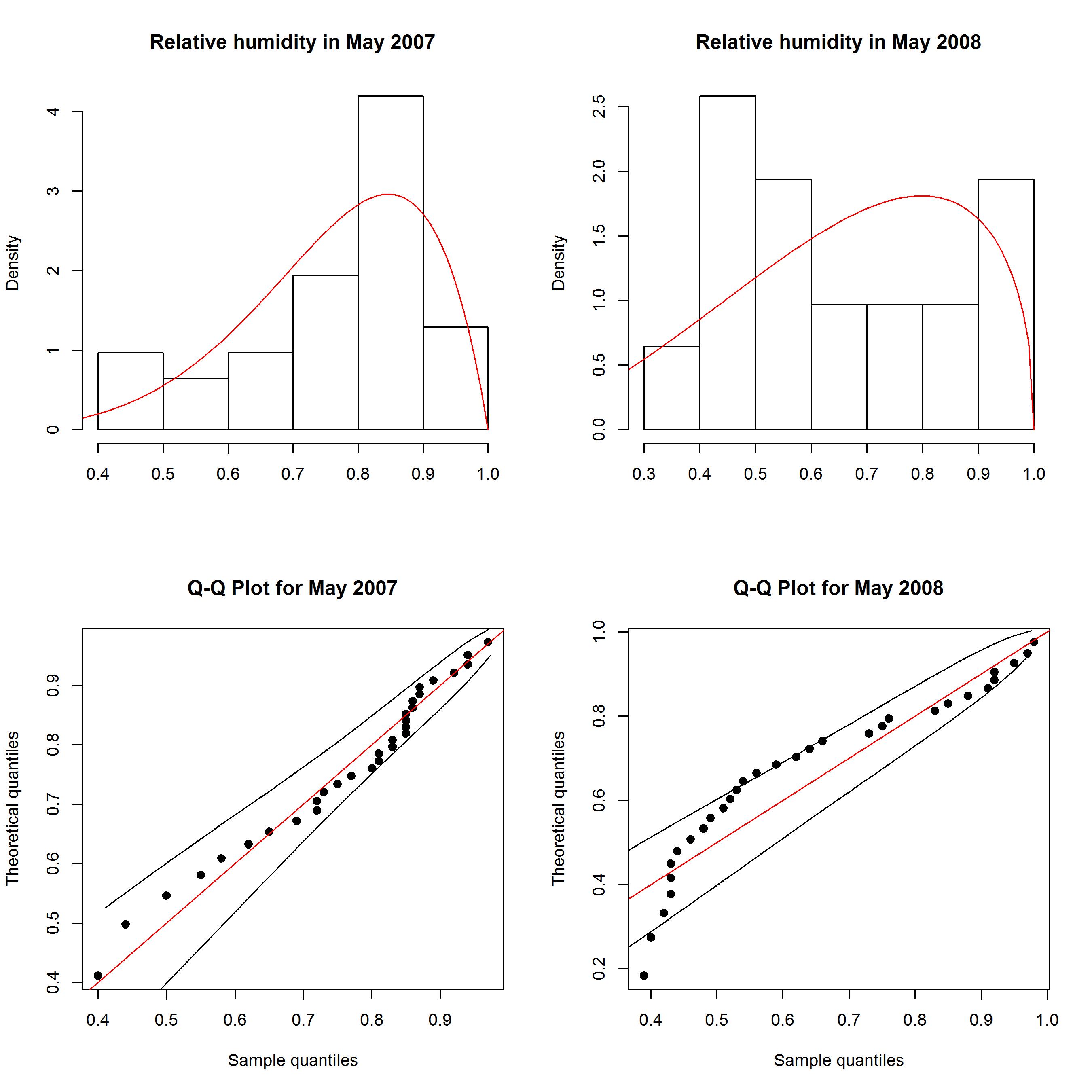}
	\caption{\label{fig:betafit} Histogram (top) and Q-Q plots (bottom) for relative humidity data.}
	\end{figure}
\begin{table}[t]
\caption{Relative humidity data of air in May 2007 and 2008 from Haarweg Wageningen weather station (see, \protect\cite{R:2011})}
\begin{minipage}{.48\textwidth}
  \centering
  \caption*{Relative humidity data for May 2007}
    \begin{tabular}{cccccccc}
    0.40  & 0.44  & 0.50  & 0.55  & 0.58  & 0.62  & 0.65  & 0.69 \\
    0.72  & 0.72  & 0.73  & 0.75  & 0.77  & 0.80  & 0.81  & 0.81 \\
    0.83  & 0.83  & 0.85  & 0.85  & 0.85  & 0.85  & 0.86  & 0.86 \\
    0.87  & 0.87  & 0.89  & 0.92  & 0.94  & 0.94  & 0.97  &  \\
    \end{tabular}%
  \label{tab:realdat1}%
\end{minipage}\hfill
\begin{minipage}{.48\textwidth}
  \centering
  \caption*{Relative humidity data for May 2008}
    \begin{tabular}{rrrrrrrr}
    0.39  & 0.4   & 0.42  & 0.43  & 0.43  & 0.43  & 0.44  & 0.46 \\
    0.48  & 0.49  & 0.51  & 0.52  & 0.53  & 0.54  & 0.56  & 0.59 \\
    0.62  & 0.64  & 0.66  & 0.73  & 0.75  & 0.76  & 0.83  & 0.85 \\
    0.88  & 0.91  & 0.92  & 0.92  & 0.95  & 0.97  & 0.98  &  \\
    \end{tabular}%
  \label{tab:realdat2}%
\end{minipage}
\end{table}

We thus apply the various tests of fit to the two data sets to investigate their behaviour on real world data. Critical values were obtained by calculating the tests on the data, and then bootstrap \textit{p}-values were obtained by first generating $10 \ 000$ samples from the beta distribution with the estimated parameters. Tests were calculated on each of the samples and the number of times the test statistic was greater than the critical value were recorded. Dividing this value by the total number of samples gives the bootstrap \textit{p}-value. The bootstrap \textit{p}-values for the two data sets are displayed in Table \ref{tab:realdatpval}. We see that the null hypothesis that the data is beta distributed is not rejected by any of the tests for the May 2007 data. However, each of the tests rejects the null hypothesis for the May 2008 data set. This observation confirms the information gathered from the plots in Figure \ref{fig:betafit}.
\begin{table}[t]
  \centering
  \caption{Bootstrap $p$-values for relative humidity data}
    \begin{tabular}{lccccccc}
    \hline
          & $KS_n$ & $CM_n$ & $AD_n$ & $RF_{n,0.25}$ & $RF_{n,2}$ & $RF_{n,5}$ & $T_n$ \\
    \hline
    May 2007 & 0.306 & 0.211 & 0.305 & 0.378 & 0.364 & 0.353 & 0.350 \\
    May 2008 & 0.024 & 0.003 & 0.002 & 0.041 & 0.044 & 0.036 & 0.007 \\
    \end{tabular}%
  \label{tab:realdatpval}%
\end{table}%

\section{Comments and Outlook}\label{sec:CaO}
In the previous sections we have shown that our testing procedure is competitive to the already known procedures. Interestingly, some characterisations of the beta law cannot be used to construct similar testing procedures, since they do not hold for all parameter combinations of $\alpha$ and $\beta$, for a specific example of such limitations see Example 13 in \cite{BE:2019}.

We finish the paper by stating some comments and giving some outlook for further research. From first calculations using the computer algebra system Maple, \cite{Maple2019}, we think that the calculation of an explicit representation of the covariance kernel in (\ref{eq:CovK}) might be possible. An explicit formula can open ground to numerical approximation of the eigenvalues of the integral operator $\mathcal{K}$ and would offer more theoretical insights. The approximation of the largest eigenvalue is, as is well known, useful for efficiency statements in the sense of Bahadur, see \cite{B:1960} and \cite{N:1995}. More insights may provide the asymptotic distribution of the test statistic under fixed alternatives. A step in this direction is given by Theorem 1 in \cite{BEH:2017}, which leads to the conjecture that
\begin{equation*}
\sqrt{n}\left(\frac{T_{n}}{n}-\Delta_{\alpha,\beta}\right)\cd \mbox{N}(0,\sigma^2(\alpha,\beta)), \quad \mbox{as}\,n\rightarrow\infty,
\end{equation*}
where $\sigma^2(\alpha,\beta)>0$ is a specified limiting variance.
Note that for the four parameter beta distribution with support on an interval $[a,b]$ (for a definition, see p.210 of \cite{JKB:1995}) and i.i.d. random variables $Y_1,\ldots,Y_n$, one can transform the data by
\begin{equation*}
X_j=\frac{Y_j-a}{b-a},\quad j=1,\ldots,n,
\end{equation*}
to a two parameter beta distribution with support in $[0,1]$. Hence, $T_n$ in (\ref{teststatTn}) can also be used to test the fit to a four parameter beta distribution if the support is known. A goodness of fit test for the situation of unknown support as well as unknown shape parameters of the beta distribution has not yet been investigated in the literature, in contrast to fitting a four parameter beta distribution to data, see \cite{AHW:1994}. We point out that the characterization given in Theorem 1 of \cite{A:1991} covers this family, but parameter estimation of the boundary points $a$ and $b$ is tricky, especially if linear expansions are needed for asymptotic theory, see \cite{W:1971} for a starting point.

\section*{Acknowledgement}
The authors thank Norbert Henze for fruitful comments and Fabio Rodrigues for the implementation of the test statistic $RF_{n,a}$ in \texttt{R}.

\bibliographystyle{abbrv}
\bibliography{lit-AC}  

%

\end{document}